\documentclass{llncs}

\usepackage{xspace,amsmath,url}
\usepackage{color}
\usepackage{booktabs}

\newcommand{\exclude}[1]{}

\newcommand{\q}{\phantom0}

\newcommand{\qqq}{\q\q\q}

\newcommand{\bis}{{\prime\prime}}

\newcommand{\THR}{\ensuremath{\mathit{THRESH}}}
\newcommand{\THRb}{\ensuremath{\mathit{THRESH_{\text{bin}}}}}
\newcommand{\ML}{\ensuremath{\mathit{MATCHLIST}}}
\newcommand{\MLb}{\ensuremath{\mathit{MATCHLIST_{\text{bin}}}}}
\newcommand{\pcpred}{\ensuremath{\mathrm{pred}}}
\newcommand{\pcsucc}{\ensuremath{\mathrm{succ}}}
\newcommand{\pcrank}{\ensuremath{\mathrm{rank}}}

\newcommand{\pcselect}{\ensuremath{\mathrm{select}}}
\newcommand{\pccopy}{\ensuremath{\mathrm{copy}}}
\newcommand{\pcinsert}{\ensuremath{\mathrm{insert}}}
\newcommand{\pcdecrease}{\ensuremath{\mathrm{decrease}}}

\newcounter{lnoc}
\newenvironment{sdalgorithm}[1]{%
\hrule height 0.8pt \vspace{0.6ex} \small#1\vspace{0.6ex}\hrule height 0.5pt \vspace{-2.0ex}
\setcounter{lnoc}{0}
\small
\begin{tabbing}
000000\=XXI\=XXI\=XXI\=XXI\=XXI\=XXI\=\kill
}{%
\end{tabbing}
\vspace{-2.0ex}\hrule height 0.8pt\vspace{0ex}}
\newcommand{\lno}[1][0]{{\footnotesize\sffamily 
\ifnum#1=0
\stepcounter{lnoc} 
\ifnum\thelnoc<10
\phantom0%
\fi
\thelnoc
\else
\thelnoc.#1
\fi
}\>}

\newcommand{\pcfor}{{\bfseries for~}}

\newcommand{\pcforeach}{{\bfseries for each~}}
\newcommand{\pcto}{{\bfseries to~}}

\newcommand{\pcdo}{{\bfseries do~}}
\newcommand{\pcif}{{\bfseries if~}}
\newcommand{\pcthen}{{\bfseries then~}}

\newcommand{\pcreturn}{{\bfseries return~}}

\newcommand{\pccomment}[1]{\qqq\{\textit{#1}\}}

\begin{document}

\title{Efficient algorithms for the longest common subsequence in $k$-length substrings}

\author{Sebastian Deorowicz$^\dag$ and Szymon Grabowski$^\ddag$}
\institute{$^\dag$ Institute of Informatics, Silesian University of Technology, \\
  Akademicka 16, 44--100 Gliwice, Poland \\
  $^\ddag$ Lodz University of Technology, Institute of Applied Computer Science,\\
  Al.\ Politechniki 11, 90--924 {\L}\'od\'z, Poland, 
  \email{sgrabow@kis.p.lodz.pl}
}
  


\maketitle

\begin{abstract}
Finding the longest common subsequence in $k$-length substrings (LCS$k$)
is a recently proposed problem motivated by computational biology.
This is a generalization of the well-known LCS problem in which 
matching symbols from two sequences $A$ and $B$ are replaced with matching non-overlapping substrings of length $k$ from $A$ and $B$.
We propose several algorithms for LCS$k$, being non-trivial 
incarnations of the major concepts known from LCS research 
(dynamic programming, sparse dynamic programming, tabulation).
Our algorithms make use of a linear-time and linear-space 
preprocessing
 finding the occurrences of all the 
substrings of length $k$ from one sequence in the other sequence.
\end{abstract}

\section{Introduction}
In last years 
the famous longest common subsequence problem~\cite{CHL2007} gave rise 
to many related sequence similarity problems, often motivated 
by computational biology.
One of them, proposed very recently by Benson et al.~\cite{BLS2013}, 
is the {\em longest common subsequence in $k$-length substrings} problem, 
which can be defined as follows.
Given two sequences, $A = a_1 a_2 \ldots a_n$ and 
$B = b_1 b_2 \ldots b_n$\footnote{All the algorithms presented in this paper 
can easily be translated to the case of sequences of arbitrary lengths 
$n$ and $m$, 
but we use the original problem definition.}
over a common alphabet $\Sigma$, 
the task is to find the maximal $\ell$ such that there exist 
$\ell$ pairs of substrings of length $k$ (called $k$-strings), 
$a_{i_e - k + 1} \ldots a_{i_e}$ and 
$b_{i_e - k + 1} \ldots b_{i_e}$, 
$1 \leq e \leq \ell$, 
where $a_{i_e - k + 1} \ldots a_{i_e}$ is equal 
to $b_{i_e - k + 1} \ldots b_{i_e}$  
and $i_e + k\leq i_{e+1}$ for any valid $e$ 
(that is, the strings of length $k$ taken from 
one of the sequences are non-overlapping).
We will often use an alternative notation for a substring: 
instead of $a_i\ldots a_j$ ($b_i\ldots b_j$) we will write 
$A_{i\ldots j}$ ($B_{i\ldots j}$).

We begin with a critique of the result from Benson et al.~\cite{BLS2013}.
The authors claim their time complexity to be $O(n^2)$, 
while in fact it is $O(k n^2)$, because 
comparing two $k$-strings
takes (na\"{\i}vely) $O(k)$ time.
In the proof of Theorem~1 they say:
``We assume that $k$ is rather a small constant thus computing $\mathit{kMatch}(i, j)$ 
is done in constant time'', which cannot be justified on a theoretical ground 
(on the other hand, their space complexity is justly presented as $O(nk)$).

We first give a (simple) fix to the technique of Benson et al., 
obtaining the true $O(n^2)$ time complexity, and then show 
three more advanced algorithms.
The first of them is based on the Hunt--Szymanski~\cite{HS1977} approach  
(originally used for the LCS problem), applying the sparse 
dynamic programming paradigm. 
The second works better if the number of matches in the dynamic programming 
matrix is large and uses the observation that matches forming a longest 
common subsequence must be separated with gaps of size at least $k$.
Its variant based on the van Emde Boas tree~\cite{vEB1977} is also briefly discussed.
Finally, a tabulation-based algorithm is presented, 
with a logarithmic speedup over the quadratic-time dynamic programming algorithm.
Our results are summarized in Table~\ref{tab:results}.

\begin{table}[t]
\caption{Our results. The last column is the complexity of the extra space needed to extract a longest common subsequence. 
The extra time for this stage is not presented, but its complexity never exceeds the corresponding time complexity to find 
the subsequence length. 
Notation: $r$ is the number of matches, $\ell \leq n/k$ is the solution length.\looseness=-1}
\label{tab:results}
\begin{footnotesize}
\begin{tabular}{lccc}
\toprule
Algorithm			& Time complexity	& Space complexity & Extraction space \\
\midrule
DP (Sect.~2)  & $O(n^2)$        & $O(nk)$     & $O(n^2)$ \\
Sparse (Sect.~3)       & $O(n + r\log\ell)$        & $O(n + \min(r, n\ell)))$     & $O(r)$ \\
Dense  (Sect.~3)        & $O(n^2 / k + n(k\log n)^{2/3})$  & $O(n)$     & $O(n\ell)$ \\
Dense-vEB (Sect.~3)    & $O(n^2\log\log n / k)$        & $O(n\log\log n)$   & $O(n\ell)$ \\
DP-4R (Sect.~4)        & $O(n^2/\log n)$ & $O(n + nk/\log n)$     & $O(n^2/\log n)$ \\
\bottomrule
\end{tabular}
\end{footnotesize}
\end{table}

\section{The LCS$k$ in $O(n^2)$ time}
\label{sec:orig}
The cornerstone for any dynamic programming (DP) based solution for 
the LCS$k$ problem will be the following recurrence.
(It is closely related to the one given by Benson et al.
We decided to introduce our own one, with match reporting
at the end rather than start symbol of the $k$-string,
since it simplifies the formulation of the algorithms 
in the rest of the paper.) 
\begin{equation}
M(i, j) = 
  \begin{cases}\max
	\begin{cases}
   M(i,j-1),\\
   M(i-1,j),\\
	\end{cases} & \text{if } A_{i-k+1\ldots i} \neq B_{j-k+1\ldots j},\\
   M(i-k,j-k) + 1, & \text{if } A_{i-k+1\ldots i} = B_{j-k+1\ldots j},\\
  \end{cases}
\label{eq:DP}
\end{equation} 
and the boundary conditions:
$M(i, j) = 0$ for all valid $i$, $j$ when $i < k$ or $j < k$.
Any location $(i, j)$ in $M$ will be called a match 
if $A_{i-k+1 \ldots i} = B_{j-k+1 \ldots j}$. 

Efficient computation of the recurrence~(\ref{eq:DP}) depends on
quick tests if $A_{i-k+1\ldots i} = B_{j-k+1\ldots j}$.
This can be achieved with the longest common extension (LCE)
query, which can be performed in $O(1)$ time after $O(n)$-time
preprocessing.
This procedure builds an augmented suffix tree for solving
the lowest common ancestor (LCA) queries in constant time~\cite{BFC2000},
over the concatenated sequence $A\#B$,
where $\#$ is a unique symbol (lexicographically largest)
working as a terminator of $A$.
Testing if $A_{i-k+1\ldots i} = B_{j-k+1\ldots j}$ translates to
the question if $\mathrm{LCE}_{A\#B}(i- k + 1, n + 1 + j - k + 1) \geq k$.

We however propose an alternative $O(n)$-time preprocessing routine
letting us access the successive matches to each
$k$-string $A_{j-k+1\ldots j}$ in sequence $B$ in constant time,
and requiring $O(n)$ words of space.
Since one list of matches is never longer than a row in the DP
matrix, we can scan the list of matches when processing each row
in a linear time,
which results in overall $O(n^2)$ time for the matrix computation.
The longest sequence itself may be extracted in a similar manner 
as in the DP algorithm for LCS, in $O(n)$ extra time and using $O(n^2)$ 
extra space.\looseness=-1

Our preprocessing routine will also be used in 
the two algorithms described 
in Section~\ref{sec:HS} 
and the algorithm from Section~\ref{sec:4R}.
The procedure makes use of a suffix array for the concatenated 
sequence $B\#A$.
We also build its longest common prefix (LCP) table; both operations 
can be accomplished in linear time and using linear space 
(see, e.g.,~\cite{KSB2006,KLAAP2001}).
The computed LCP values allow us to partition the sorted set of suffixes 
into maximal groups such that the LCP between successive items is at least $k$.
In other words, suffixes from such groups have a prefix of length $k$ symbols
in common.

These $k$-string groups are radix-sorted according to the 
starting position of the suffix. 
To make it efficient ($O(n)$ time), the sort is performed once for all groups;
all the suffixes are represented as pairs $(\mathit{group\_id}, \mathit{start\_pos})$, 
where $\mathit{group\_id}$ is 1 for the first group, 2 for the second group, etc., 
in their position order.
Let us denote the array with such pairs with $S$.
After the sort, suffixes in the groups are kept together, in 
starting position order.
Note that within a group all suffixes starting in $B$ are located before 
any suffix starting in~$A$.

We scan over all the items in $S$, except for the last one 
(which must correspond to the suffix starting with $\#$), 
and we insert related data into another array, $X$ of length $|A|+|B| = 2n$. 
More precisely, for each examined $S[i]$ we write in $X[S[i].\mathit{start\_pos}]$ 
a triple:
$(\mathit{fa\_pos}, \mathit{fb\_pos}, \mathit{ng\_pos})$,
where $\mathit{fa\_pos}$ ($\mathit{fb\_pos}$) is the position in $S$ of the first suffix 
from this group starting in $A$ (in $B$)
and $\mathit{ng\_pos}$ is the position in $S$ of the first suffix from the next group.
This operation also takes linear time and the preprocessing is done.

As stated above, a rowwise scan requires fast access to all matches to
successive 
$A_{i-k+1 \ldots i}$ $k$-strings.
It is now enough to examine the information stored at $X[S[i].\mathit{start\_pos}]$ 
(which in turn refers to $S$), to find the match locations in the row 
in $O(1)$ time per each.

Below we 
give
two simple properties of the matrix $M$.

\begin{lemma}
For each $i$ and $j$ the value $M(i, j)$ is the LCS$k$ of the prefixes $A_{1\ldots i}$ and $B_{1\ldots j}$.
\end{lemma}


The proof is rather straightforward and is very similar to the classic one 
for the LCS problem~\cite{CHL2007}.
As a consequence, $M(n, n)$ is the solution of the LCS$k$ problem.
It is also easy to notice that 
$M(i,j+1) - M(i,j) \in \{0, 1\}$ 
and $M(i+1,j) - M(i,j) \in \{0, 1\}$ for all valid $i$ and $j$.
The following lemma describes a feature of the matrix $M$ 
which will be crucial for both 
algorithms from Section~\ref{sec:HS}.

\begin{lemma} \label{lemma2}
Let vector $V(i)$, for any $i$, describe the changes in $i$th row of~$M$, i.e., $V(i)$ stores the pairs $\langle M(i, j), j\rangle$ such that $M(i,j-1) + 1 = M(i,j)$.
Then, each $\langle h, j\rangle \in V(i)$ implies that it is impossible that $\langle h+1, j^\prime\rangle \in V(i^\prime)$, for any $i \le i^\prime \le i+k$ and $j^\prime < j+k$.
\end{lemma}

\begin{proof}
Let us assume otherwise, so let $\langle h+1, j^\prime\rangle \in V(i^\prime)$, for some $i \le i^\prime \le i+k$ and $j^\prime < j+k$.
It means that $M(i^\prime, j^\prime) = h+1 = 
\mathit{LCSk} (A_{1\ldots i^\prime}, B_{1\ldots j^\prime})$.
Truncating both 
sequences by $k$ symbols does not change their LCS$k$ or reduces it by~1, so
$M(i^\prime-k, j^\prime-k) \ge h$.
This is, however, impossible as $i^\prime-k \le i$ and $j^\prime-k < j$ and the leftmost cell of $M$ among rows $0, \ldots, i$ containing value $h$ is in column~$j$.
\qed
\end{proof}

Simply speaking, the above lemma says that the increments in $M$ are separated by at least $k$ cells in both horizontal and vertical directions.

\section{Two sparse dynamic programming algorithms}
\label{sec:HS}

One of the major paradigms for solving LCS-related problems is 
{\em sparse dynamic programming} (SDP).
The overall idea is to visit only those DP matrix cells which correspond to matches.
As the number of matches, $r$, is usually significantly 
smaller than $n^2$, we can often expect 
a significant speedup over the standard DP procedure.
The first such algorithm for the LCS problem was given by Hunt and Szymanski~\cite{HS1977}, 
with 
$O(r\log\ell)$ time in its basic variant, 
where $\ell \leq n$ is the LCS length.
While this complexity is superquadratic in the worst case 
(i.e., for $r = \Theta(n^2)$),
there exists a theoretical solution based on the HS approach which is free 
of this drawback~\cite[Sect.~5]{EGGI1992}. 
In this section we will present two SDP algorithms for the LCS$k$ problem, 
the first of which being
an adaptation of the HS approach.
This algorithm, called LCS$k$-Sparse, is presented in Fig.~\ref{fig:lcs:alg}.

\begin{figure}[t]
\begin{sdalgorithm}{LCS$k$-*($A$, $B$, $k$)}
\lno	Compute $\ML$ for $A$ and $B$\\
\lno	$\THR[0] \gets \langle -\infty, +\infty \rangle$\\
\lno	\pcfor $i \gets 1$ \pcto $k-1$ \pcdo\\
\lno\>	$\THR[i] \gets \pccopy(\THR[i-1])$\\
\lno	\pcfor $i \gets k$ \pcto $n$ \pcdo\\
\lno\>	$\THR[i] \gets \pccopy(\THR[i-1])$\\
\lno\>	$\sigma_{A} \gets \text{get\_k-string}(A, i)$\\
\pccomment{Sparse variant}\\
\lno\>	\pcforeach match $x$ in $\ML[\sigma_A]$ \pcdo\\
\lno\>\>		$j^\prime \gets \THR[i-k].\pcpred(x-k+1)$\\
\lno\>\>		$h \gets \THR[i-k].\pcrank(j^\prime)$\\
\lno\>\>		$j^\bis \gets \THR[i].\pcselect(h+1)$\\
\lno\>\>		\pcif $x < j^\bis$ \pcthen\\
\lno\>\>\>			$\THR[i].\pcdecrease(j^\bis, x)$\\
\lno\>\>\>			\pcif $j^\bis = +\infty$ \pcthen\\
\lno\>\>\>\>			$\THR[i].\pcinsert(+\infty)$\\
\lno	\pcreturn $|\THR[n]|-2$\\
\pccomment{Dense variant\setcounter{lnoc}{7}}\\
\lno\>	\pcfor $h \gets 1$ \pcto $|\THR[i-1]|-1$ \pcdo\\
\lno\>\>		$j^\prime \gets \THR[i-k][h-1]$\\
\lno\>\>		$x \gets \ML[\sigma_A].\pcsucc(j^\prime+k-1)$\\
\lno\>\>		$j^\bis \gets \THR[i][h]$\\
\lno\>\>		\pcif $x < j^\bis$ \pcthen\\
\lno\>\>\>		$\THR[i][h] \gets x$\\
\lno\>\>\>			\pcif $j^\bis = +\infty$ \pcthen\\
\lno\>\>\>\>			$\THR[i].\pcinsert(+\infty)$\\
\lno	\pcreturn $|\THR[n]|-2$
\end{sdalgorithm}
\caption{The sparse and dense DP algorithms for the LCS$k$ problem.}
\label{fig:lcs:alg}
\end{figure}


Assume that we are going to visit the matches rowwise, 
each row scanned from left to right.
We start with a simple definition.
$M(i, j)$ 
stores a match of {\em rank} $h$ iff $A_{i-k+1\ldots i} = B_{j-k+1\ldots j}$  and $\mathrm{LCSk}(A_{1 \ldots i}, B_{1 \ldots j}) = h$.
In the preprocessing (line~1), lists of successive occurrences of all 
$k$-strings from the sequence~$A$ are gathered (in $O(n)$ time), 
as described in Section~\ref{sec:orig}.
The main processing phase makes use of a persistent red-black 
tree~\cite{DSST1989} $\THR$ for maintaining 
the leftmost seen-so-far columns of matches of growing ranks.
More precisely, accessing $\THR[i][h]$ answers the question 
about the leftmost column in row $i$ with a match of rank $h$.
When processing row $i$, we will often be interested in accessing 
$\THR[i-k]$, i.e., the state of this structure $k$ rows earlier.
For the first $k$ rows of the DP matrix the structure $\THR$
contains only two sentinel values, $-\infty$ and $+\infty$ (lines~2--4), 
the former with rank 0.
For each of the following rows, the state of $\THR$ 
from the previous row is modified only if the current match 
on $\ML$, of rank $h+1$, is in column $x$, which is less 
than the column $j''$ of the $(h+1)$-th value in $\THR[i-1]$ 
(i.e., also $\THR[i]$ so far).
This modification (lines~12--13) involves decreasing a value in the 
structure, which may be implemented as one delete and one insert 
operation.
Note that when the decreased value is the $+\infty$ sentinel, 
the $\THR[i]$ grows by one ($+\infty$ is again inserted at its 
end, in line 15).
All operations on $\THR$ have logarithmic cost in the worst case, 
that is, $O(\log\ell)$, where $\ell \leq n/k$ is the LCS$k$ length.
The overall worst case time for the algorithm is thus $O(n + r\log\ell)$.
The space consumption is usually determined by the number of matches $r$
(we need $O(1)$ nodes of the persistent RB tree per each match).

The presented code only returns the length of a LCS$k$.
Yet, to obtain the common subsequence itself we only have to modify 
the algorithm slightly, and with each entry in $\THR[i]$ store 
a reference to the $\THR$ value used to compute the current value. 
This enables backtracking the solution in $O(\ell)$ extra time 
and using $O(r)$ words of space.



If the number of matches is close to $n^2$, 
a better solution is to use the algorithm LCS$k$-Dense 
(Fig.~\ref{fig:lcs:alg}).
The first steps (lines~1--4) resemble the previous variant, 
but here the data structure $\THR$ is not persistent 
(and may be simply a dynamic array), 
hence the copy routine, used for each row $i$, has its cost linear in 
the size of $\THR[i-1]$.
The main loop, which is run for each row $i$, $i \geq k$, 
starts with making a copy of $\THR[i-1]$ into $\THR[i]$.
Then, $\THR[i]$ is traversed 
in order, and its $h$-th value updated based on the 
current $\ML$ and the $\THR$ in its state $k$ rows earlier.
More precisely, if $(h-1)$-th value of $\THR[i-k]$ 
is denoted with $j'$ (line~9)
and the nearest match on the current $\ML$ in a column greater 
or equal to $j' + k$, denoted by $x$ (line~10), is less than the current 
$h$-th element of $\THR$ (line~11), then $\THR[i][h]$ is 
updated to $x$ (lines~12--13).
As in the previous algorithm, $\THR[i]$ may get longer by one  
(line~15).
Finding the LCS$k$ length needs access to $k$ previous rows of $\THR$, 
and as each of them 
contains at most $\ell + 2 \leq n/k + 2$ items, 
the overall space is $O(n)$.
Also in this algorithm the desired sequence may be backtracked 
in $O(\ell)$ extra time, 
if backlinks are stored together with $\THR$ entries.
The memory use, however, is here $O(n\ell)$ words of space, 
due to physical copying of the $\THR[i]$ structures.

Let us analyze the time complexity of this algorithm.
It depends on how efficiently we can handle the successor queries 
in line~10.
Binary search over $\ML[\sigma_A]$
gives a factor $\log n$.
If $k$ is small enough, however, 
we can remove the logarithmic factor.
Two rows from the DP matrix will be considered equal 
(or one called a duplicate of another)
if they have matches in the same set of columns.
We start with a simple observation: for any $q \geq 1$ there cannot 
be more than $q$ {\em distinct} rows with at least $n/q$ matches 
in each of them.
Let us use two thresholds, $t_1$ and $t_2$, where $1 \leq t_2 < t_1 < n$.
We set $t_1 = k\log n$ and let the rows with less than  
$n/t_1$ matches be called ``sparse'',
those with at least $n/t_1$ and at most $n/t_2$ matches (the exact value 
of $t_2$ will be found later) be called ``dense'', 
and finally those rows with more than $n/t_2$ matches be called ``superdense''.
In the sparse rows, we calculate the successor query in $O(\log n)$ time, 
spending $O(n^2\log n/t_1) = O(n^2/k)$ time in total for them.

For the dense blocks, we partition each row into $n/b$ blocks 
of size $b$ cells each, where the exact value of $b$ will be found later.
Let us focus on a dense row $i$.
For each block $M[i][j+1 \ldots j+b]$ we first find 
$\ML[\sigma_A].\pcsucc(j+1+k-1)$ and $\ML[\sigma_A].\pcsucc(j+b+k-1)$ 
(using a linear scan over $\ML[\sigma_A]$) 
and if both values are the same, it means that all the cells 
in this block have the same successor used in line~10.
If not, we associate with this block a list of all its $b$ successors, 
one per each element from the block.
These values are stored as dynamic arrays, one per block, 
of size 1 or $b$.
The total time spent per a dense row is $O(n/b + nb/t_2)$.
There are at most $t_1 = k\log n$ distinct dense rows, 
and finding the successors for all of them takes 
$O(k\log n (n/b + nb / t_2))$ time, minimized for $b = \sqrt{t_2}$ to $O(nk\log n / \sqrt{t_2})$
(duplicate rows obtain references to the already computed answers, 
in $O(n)$ total time).

Finally, superdense rows are processed na{\"i}vely in $O(n)$ time each, 
with $O(n t_2)$ total time.
Overall, we obtain $O(n^2 / k + n k\log n / \sqrt{t_2} + n t_2)$ time, 
which is minimized for $t_2 = (k\log n)^{2/3}$, to yield 
$O(n^2 / k + n(k\log n)^{2/3})$ time.
This reduces to simply $O(n^2 / k)$ as long as 
$k = O((n / (\log n)^{2/3})^{3/5})$.

Alternatively, the successor queries may be handled with the 
famous van Emde Boas (vEB) tree~\cite{vEB1977}, in $O(\log\log n)$ time.
We need to maintain $O(n)$ such structures, 
using a variant with lazy initalization.
In this way, the total time of the insertions (including initializations) is 
$O(n \log \log n)$ and so is the space consumption.
The overall time complexity of this variant is $O(n^2\log\log n / k)$ for any~$k$.

\section{Tabulation-based algorithm}
\label{sec:4R}

The tabulation (also called ``4-Russians'') technique for 
dynamic programming algorithms 
consists in dividing the DP matrix into small blocks 
(usually $1 \times b$ or $b \times b$, for some $b$), 
such that the number of distinct blocks is small enough to be 
precomputed beforehand, e.g., with linear time-space resources.
For the LCS problem this technique was first applied by 
Masek and Paterson~\cite{MP1980}.

Let us now present a tabulation-based algorithm for LCS$k$, 
called DP-4R; 
the reader needs to know the (purpose of the) data structures
$\THR$ and $\ML$ used in the previous section.
In DP-4R, the current state of the list $\THR$ is represented 
as a bit-vector $\THRb$ of length $n$
and similarly the match lists, $\MLb$, for all $k$-strings from $A$ are built 
(to avoid $O(n^2)$ bits of space in the worst case, these lists 
can be built on the fly, one per row).
More precisely, $\THRb[i][j] = 1$ iff $M(i,j) - M(i,j-1) = 1$, for any 
$1 \leq j \leq n$.
Similarly, $\MLb[\sigma_A][j] = 1$ iff $B_{j-k+1 \ldots j} = \sigma_A$.
For the current row $i$, $i \geq k$, each snippet $\THRb[i][j+1 \ldots j+b]$ 
depends only on:\looseness=-1
\begin{itemize}
  \item the snippet $\THRb[i-1][j+1 \ldots j+b]$,
  \item the snippet $\THRb[i-k][j-k+1 \ldots j-k+b]$,
  \item the snippet $\MLb[\sigma_A][j+1 \ldots j+b]$,
  \item the difference $M(i, j) - M(i-1, j) \in \{0, 1\}$,
  \item the difference $M(i, j) - M(i-k, j-k) \in \{0, 1\}$.
\end{itemize}
(Both listed differences can be tracked easily 
during the rowwise snippet processing.)

Now, if $b = \Theta(\log n)$ with a small enough constant, we can compute 
the current snippet of $\THRb[i]$ in constant time, with a lookup table 
built in the preprocessing (e.g., in $O(n)$ time), obtaining 
an $O(n^2/\log n)$-time algorithm.
During the computations, the previous $k$ rows of $\THRb$ of length $n$ bits 
need to be available, which makes the overall space use $O(n + nk/\log n)$ 
words.

It remains however to explain how $\MLb[\sigma_A][j+1 \ldots j+b]$ snippets 
are prepared. 
To this end, we note that all snippets from a row 
$\MLb[\sigma_A][1 \ldots n]$ can easily be created from a corresponding 
match list (found in the linear-time preprocessing) in $O(\max(n/\log n, r'))$ 
time, where $r' \leq n$ is the number of matches in this row.
This means that all 
sparse rows, i.e. such for which $r' = O(n/\log n)$,  
pose no problem as the worst-case time of creating their $\MLb$ bit-vectors 
sums to $O(n^2/\log n)$.
The number of {\em distinct} remaining (dense) rows in the matrix 
is however limited to less than $\log n$ (cf. a similar reasoning in 
Section~\ref{sec:HS} for the algorithm LCS$k$-Dense), 
hence the $O(n\log n + n^2/\log n)$ time for preparing their snippets, 
including their first occurrences and all duplicates, 
does not hamper the overall time complexity either.

An LCS$k$ sequence can now be extracted basically like in the 
plain DP approach, in $O(n + k\ell) = O(n)$ time.
To this end, the last 1 in $\THRb[n]$ is found, with a linear scan 
from right to left, and its column $j$ is the end position of the 
last $k$-string in the result.
After that, we go to the row $n-k$ and column $j-k$, 
and scan left for the nearest 1,
which will correspond to the penultimate $k$-string, and so on.

It is tempting to devise a similar algorithm based on bit logic 
rather than a precomputed table, 
but we suppose that obtaining $O(n^2/w)$ time, 
where $w \geq \log n$ is the machine word size, may be hard or even 
impossible for the LCS$k$ problem.


\section{Conclusions}

We presented four algorithms, with respectively $O(n^2)$, $O(n + r \log\ell)$,
$O(n^2 / k + n(k\log n)^{2/3})$ 
and $O(n^2/\log n)$ time complexities, 
for the recently introduced problem 
of finding the longest common subsequence in $k$-length substrings.
We used several major techniques known from the research on LCS 
and related problems: dynamic programming, sparse dynamic programming, 
tabulation.
Their application to LCS$k$ was, however, non-trivial; for example 
using the Hunt--Szymanski approach required a persistent data structure 
to preserve an attractive time complexity.


\section*{Acknowledgement}
The work was supported by the Polish National Science Center upon decision DEC-2011/03/B/ST6/01588 (first author).

\bibliographystyle{abbrv}
\bibliography{lcsk}

\end{document}